\documentclass[10pt,a4paper]{article}
\usepackage{graphicx}
\usepackage{amssymb}
\usepackage{color}
\usepackage{amsthm}
\usepackage{float}
\usepackage[T1]{fontenc}
\usepackage[utf8]{inputenc}
\usepackage{authblk}
\usepackage[bottom=0.9in,top=0.8in]{geometry}

\newtheorem{prop}{Proposition}
\newtheorem{theorem}{Theorem}
\newtheorem{lemma}{Lemma}

\title{The Maximum Number of 3- and 4-Cliques within a Planar Maximally Filtered Graph}
\author[1]{Jenna Birch \thanks{j.l.birch@student.liv.ac.uk}}
\author[1,2] {Athanasios A. Pantelous \thanks{aap@liv.ac.uk}}
\author[2] {Konstantin Zuev \thanks{zuev@liv.ac.uk}}
\affil[1] {Institute for Financial and Actuarial Mathematics, Department of Mathematical Sciences, University of Liverpool, United Kingdom.}
\affil[2]{Institute for Risk and Uncertainty, University of Liverpool, United Kingdom.}

\date{February 2014}

\begin{document}

\maketitle
\section*{Abstract}
\
Planar Maximally Filtered Graphs (PMFG) are an important tool for filtering the most relevant information from correlation based networks such as stock market networks. One of the main characteristics of a PMFG is the number of its 3- and 4-cliques. Recently in a few high impact papers it was stated that, based on heuristic evidence, the maximum number of 3- and 4-cliques that can exist in a PMFG with $n$ vertices is $3n - 8$ and $n - 4$ respectively. In this paper, we prove that this is indeed the case.

\section{Introduction}

In recent years there has been increasing interest in how we can model complex systems using network theory. These can include information, technological and biological systems \cite{Newman03}, social networks \cite{Toivonen06} and financial markets \cite{Allen00}, \cite{Bonanno04}. In particular, a network based approach of studying complex systems has become very popular in econophysics \cite{Stanley99}, an interdisciplinary research field that studies economic and financial phenomena. One of the important and fundamental problems in this approach is to filter the most relevant information from financial networks. As a result traditional algorithms from network theory have been adapted and some new methods have been introduced. In 1999 Mantegna \cite{Mantegna99} introduced a method for finding a hierarchical arrangement for a portfolio of stocks by extracting the Minimum Spanning Tree (MST) from the complete network of correlations of daily closing price returns for US stocks. This work has developed to include Asset Graphs (AG) as introduced in Refs. \cite{Onnela02}, \cite{Onnela03} and Threshold Networks \cite{Tse10}, \cite{Qiu10}. Building from the work by Mantegna \cite{Mantegna99} with the MST, one of the most recent developments was an algorithm proposed by Tumminello \textit{et al.} \cite{Tumminello05} where the complete network can be filtered at a chosen level, by varying the genus of the resulting filtered graph. So if a graph is embedded on a surface with genus = $g$, as $g$ increases the resulting graph becomes more complex and so reveals more information about the clusters formed, while keeping the same hierarchical tree as the corresponding MST. The simplest form of this graph is the Planar Maximally Filtered Graph (PMFG), on surface $g = 0$. \\

The PMFG is constructed in a similar way to the MST. Specifically, for a graph with $n$ vertices, a weighted edge is associated with all paired vertices where the value of the edge is the similarity coefficient between the two vertices. The weighted edges $u_{1}, u_{2}, ..., u_{e}$  are placed in descending order $u_{(1)}, u_{(2)}, ..., u_{(e)}$. The first edge $u_{(1)}$ is selected and a graph is constructed with $u_{(1)}$ and the two vertices that it connects. The ordered edges continue to be selected and added to the network structure only if the resulting network can be embedded into a plane i.e. can be drawn on a planar surface without edges crossing. (There are some tests for planarity based on Kuratowski's theorem. For more details on these refer to Hopcroft and Tarjan \cite{Hopcroft74}). The algorithm ends when all vertices $v_{1}, v_{2}, ..., v_{n}$ are connected, using $3(n-2)$ edges.\\

Mantegna \cite{Mantegna99} illustrates how the MST can extract the hierarchical arrangement between the vertices of a correlation network (in particular between stocks). Similar to this, one of the key properties of the AG, threshold networks and PMFG is that cliques can form between the vertices in the network which can highlight relationships. Huang \textit{et al.} \cite{Huang09}creates threshold networks to analyse the Chinese stock market using a correlation threshold value $-1 \leq \theta \leq 1$ where $\theta$  is the correlation coefficient between two stocks. They study the relationship between the maximum clique, maximum independent set (a subset $I \subseteq V$ such that the subgraph $G(I)$ has no edges) and the threshold value $\theta$.  Huang \textit{et al.} \cite{Huang09} state that \textit{`the financial interpretation of the clique in the stock correlation network is that it defines the set of stocks whose price fluctuations exhibit a similar behaviour.' }\\

The PMFG has proven to be an important tool for filtering the most relevant information from a network, particularly in correlation based networks that model the correlation between stock prices. For example, Pozzi \textit{et al.} \cite{Pozzi13} consider the level of risk and the returns on portfolio’s selected using filtered graphs, including PMFG. Erygit \textit{et al.} \cite{Eryigit09} use PMFGs (along with MSTs and clustering methods) to analyse the daily and weekly return correlations among indices from stock exchange markets of 59 countries. In general, the PMFG can tell us about the clusters (as mentioned above) that form within the dataset, regardless of the network nature, as a result of the underlying topological properties of the network. Song \textit{et al.} \cite{Song12} introduce a technique to extract the cluster structure and detect the hierarchical organisation within a complex dataset. This method has been developed using the topological structure of the PMFG such as the separating 3-cliques which separate a graph into two disconnected parts. Aste \textit{et al.} \cite{Aste05a} discuss the benefits of studying networks in terms of their surface embeddings. By considering the topology of the PMFG we can see that the basic structure (or motif) of the PMFG is a series of 3-cliques. Consider a sphere, a surface with $g=0.$ The PMFG separates the sphere into a sequence of triangular faces, with each vertex of the network belonging to a 3-clique. For a set of vertices there are various representations that this underlying series of 3-cliques can form (see Section 2.2). A set of three 3-cliques joined by the shared edges of a fourth 3-clique will form a 4-clique between a group of four vertices. Aste \textit{et al.} \cite{Aste05b} discusses that there must be strong relations between the properties of these 4-cliques and the ones of the system from which the cliques have been generated. For the PMFG we consider the vertices that form the 3- and 4-cliques (as the maximum number of elements that can form a clique is 4). Tumminello \textit{et al} \cite{Tumminello05} state \textit{`...normalizing quantities are $n_{s} - 3$ for 4-cliques and $3n_{s} - 8$ for 3-cliques. Although we lack a formal proof, our investigations suggest that these numbers are the maximal number of 4-cliques and 3-cliques, respectively, that can be observed in a PMFG of $n_{s}$ elements.'} As well as looking at the average correlations within the cliques and whether the cliques are from one sector or cross-sector we also consider the ratio between the number of cliques that have formed to the maximum number of cliques that could form. For this, Tumminello \cite{Tumminello05} used the normalizing quantities that have been mentioned above, an approach that has also been used by Refs. \cite{Eryigit09}, \cite{Aste05b}, \cite{Tumminello07} and used when defining the connection strength of a sector in Ref. \cite{Coronnello05}. This paper provides the missing formal proof that $3n - 8$ and $n - 4$ are indeed the maximum numbers of 3-cliques and 4-cliques possible in a PMFG.\\

This paper is organised as follows: Section 2 provides relational definitions and notations including diagonal flips (2.1) and surface triangles and separating 3-cycles (2.2). Section 3 discusses various representations of each maximal planar graph, including standard spherical triangulation form (3.1). Our main results are presented in Section 4. Section 5 concludes and gives future direction for our research.

\section{Relational Definitions and Notations}

Here we introduce some key terminology that we use throughout the paper. Consider a graph $G(V,E)$ where $V$ is the set of vertices belonging to $G$ and $E$ is the set of edges belonging to $G$. We denote the number of vertices $|V|=n$ and the number of edges $|E|=e$.  \\

Let $G$ be a \textit{planar graph}, i.e. a graph that can be embedded in the plane in such a way that the edges of $G$ will only intersect at the end points (the vertices of $G$). The planar graph divides the plane into \textit{faces}, with each face bound by a simple cycle of $G$. The number of edges in this boundary is the \textit{degree} of the face. The \textit{planar representations} of $G$ are all possible isomorphic embeddings of $G$ in the plane. \\

A \textit{triangulation} of a closed surface is a simple graph, one that does not contain self- or multiple-edges, which is embedded into the surface so that each face is a triangle and that two faces meet along at most one edge. A planar graph is \textit{maximal} if it is triangulated because if a face has more than 3 edges we can add a diagonal edge. A PMFG is a triangulation of a sphere. Within this paper we shall denote $P_{n}$ as a \textit{maximal planar graph} with $n$ vertices. \\

A \textit{chord} is an edge connecting two vertices of a cycle, which is not included in the cycle itself. For a graph $P_{n}$, a cycle of length $k$ $(k \geq 3)$ is called a $k$-cycle, denoted as $\mathcal{C}_k$. A cycle $\mathcal{C}$ is a \textit{pure chord-cycle} if the interior of $\mathcal{C}$ contains no vertices and all the interior faces of $\mathcal{C}$ are triangles. If each of the cycles of 4 or more vertices within a graph has a chord, then the graph is called a \textit{chordal graph}. A \textit{wheel graph}, denoted as $W_{n}$, is a graph with $n \geq 4$ formed by connecting a single vertex to all other vertices of an $(n-1)$-cycle. 

\subsection{Diagonal Flips}

Consider two triangular faces which share a common edge and form a quadrilateral, (see Figure 1). Negami \cite{Negami94} defines a \textit{diagonal flip} of an edge as replacing the existing common edge with a new edge between the other two vertices. A diagonal flip is only possible if the resulting quadrilateral does not contain any multiple edges. \\

\begin{figure}
\begin{center}
\includegraphics[width=.8\textwidth]{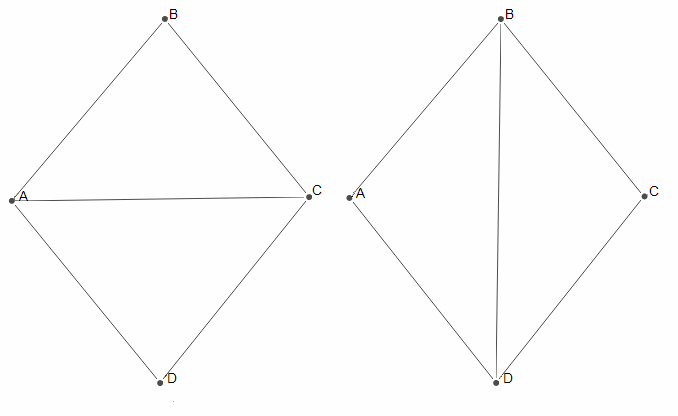} 
\caption{A quadrilateral ABCD is formed by the two adjoining triangles ABC and ACD which share a common edge (A,C). If we perform a diagonal flip the edge (A,C) is replaced by the edge (B,D).}
\end{center}
\end{figure}

\subsection{Surface Triangles and Separating 3-Cycles}

For a graph $G(V,E)$ a subset of vertices $C \subseteq V$ is called a \textit{clique} if the subgraph $G(C)$ is a complete graph and is denoted as $C_{m}$ where $|C| = m.$ As a result of Kuratowski's Theorem  we know that the PMFG allows cliques up to a maximum size of 4 vertices and so in this paper we will consider cliques of sizes 3 and 4. The 3-cliques can take the form of triangles on the surface (a pure chord-cycle of length 3 that forms a face of the PMFG) or a \textit{separating 3-cycle} (a 3-cycle where both the interior and exterior of $C_{3}$ contain one or more vertices), (see Figure 2).\\

\begin{figure}
\begin{center}
\includegraphics[width=.4\textwidth]{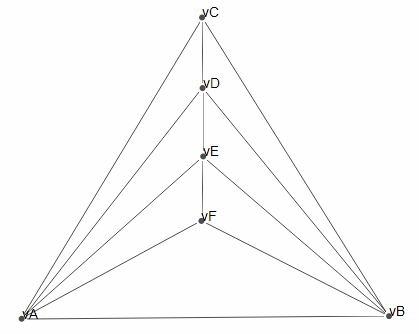}  
\caption{This PMFG with 6 vertices highlights the two possible 3-cliques. Vertices A,C,D form a 3-clique and they outline a triangle on the surface. Vertices A,B,E also form a 3-clique however they do not outline a surface triangle but rather the edges enclose 3 surface triangles which share common edges. A,B,E forms a separating 3-cycle.}
\end{center}
\end{figure}

The faces bounded by a cycle of edges are called \textit{finite faces} whereas the unbounded face (ABC in Figure 2) is called the \textit{infinite face}. As the PMFG is a triangulation of the sphere this unbounded infinite face will also form a triangle. \\

In Section 4 we will study the maximum number of 3-cliques possible in a PMFG, however we begin by studying the maximum number of triangles on the plane (i.e. the maximum number of faces). To do this we use the Handshaking Lemma and Euler's formula (refer to Ref. \cite{Wilson85} for further details). For the remaining of this section let $G(V,E)$ be a simple, undirected, finite planar graph.\\

\begin{prop} 
Let $G$ be a PMFG with $n$ vertices, $e$ edges and $f$ faces. Then $e=3n-6$ and $f=2n-4$.
\end{prop}

\begin{proof} 
Since $G$ is planar and $deg(v_{i}) \geq 2$, \\
\[\sum^{f}_{i=1} deg(f_i) = 2e.\] \\
Since $G$ is a triangulation, $deg(f_i) = 3 \Rightarrow 3f = 2e$.\\
We can substitute into Euler's formula and obtain the following,  \\

\[n - \frac{3}{2} f + f = 2\]
\begin{equation}
f = 2n - 4
\end{equation}

and similarly,\\
\[n - e + \frac{2}{3} e = 2\]
\begin{equation}
e = 3n - 6
\end{equation}
\end{proof}

So for a PMFG, the triangulation of a sphere, we have a maximum number of $2n-4$ surface triangles.

\section{ Representations of Each Maximal Planar Graph}

In this section the various representations of planar graphs are presented, which are used to achieve the main results shown in Section 4.\\

Using the relationship between the number of edges and the sum of the vertex degree we can calculate the maximum sum of all vertex degrees for a PMFG with $n$ vertices. By considering all combinations of the possible degrees of each vertex we can see what embeddings would be possible and, from these, which would be planar graphs. The following worked example shows this more clearly.\\

\textbf{Example 1.} Take $G(V,E)$ where $|V|=8$ then we know:\\
$e = 3n-6 = 18$ and so $\sum^8_{i=1} deg(v_i )=2e=36.$\\

Then each vertex can have a degree value from the set \{3, 4, 5, 6, 7\} due to the restrictions that each vertex can be joined to all other vertices at most once (the PMFG does not allow for multiple edges) and the degree of each vertex must be at least 3. From this set there are 53 possible combinations that would give the total degree sum of 36 and from these combinations 13 would produce a planar graph. Graphs that have the same combination of vertex degrees and are isomorphic to each other are known as \textit{planar representations} and they will have the same number of 3-cliques (denoted $C_3$) and 4-cliques (denoted $C_4$), (see Figure 3). \\

\begin{figure}[H]
\begin{center}
\includegraphics[width=.8\textwidth]{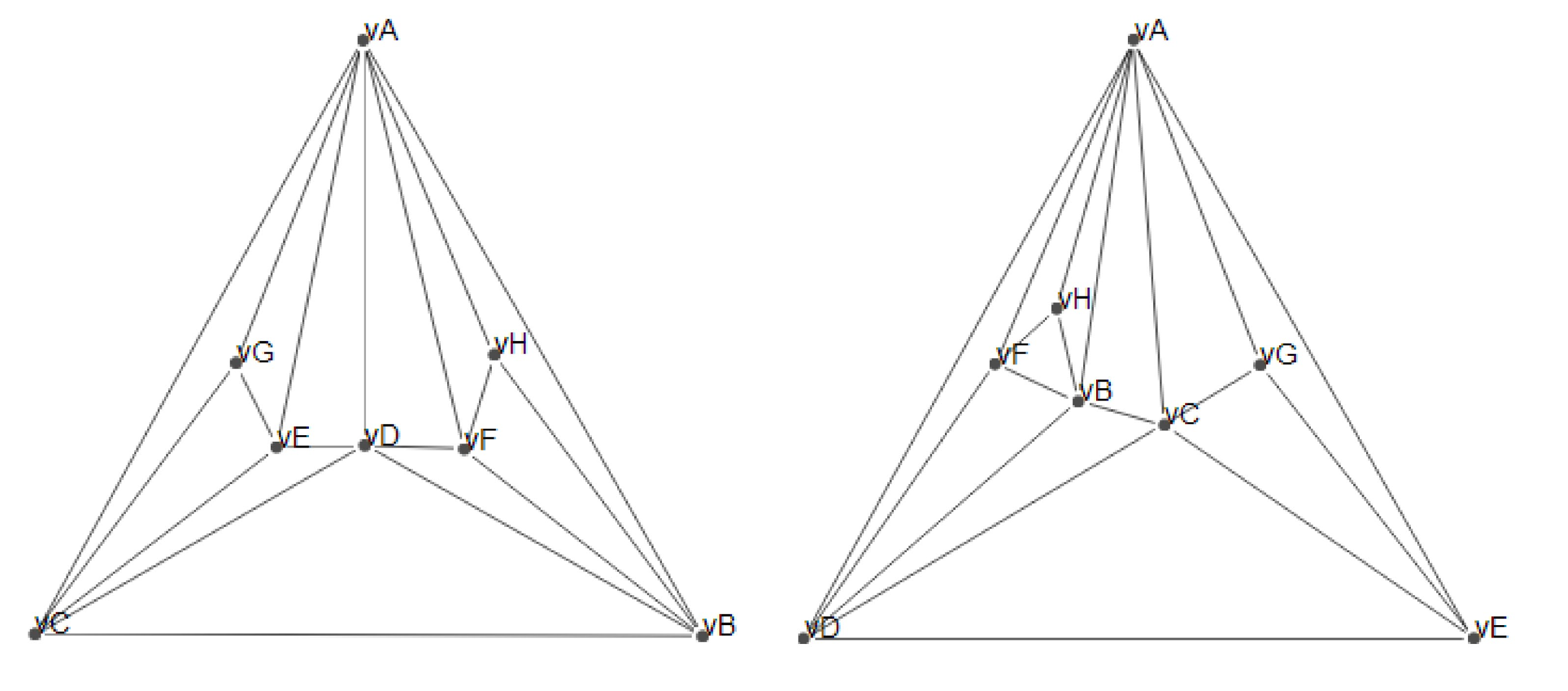} 
\caption{PMFGs with $n=8$ and $deg[v_{A}, v_{B}, v_{C}, ..., v_{H}]$=[7, 5, 5, 5, 4, 4, 3, 3]. These graphs are isomorphic and have $C_3=16$ and $C_4=5$.}
\end{center}
\end{figure} 

It is possible however to have graph structures with the same combination of vertex degrees that are not isomorphic, (see Figure 4).

\begin{figure}
\begin{center}
\includegraphics[width=.8\textwidth]{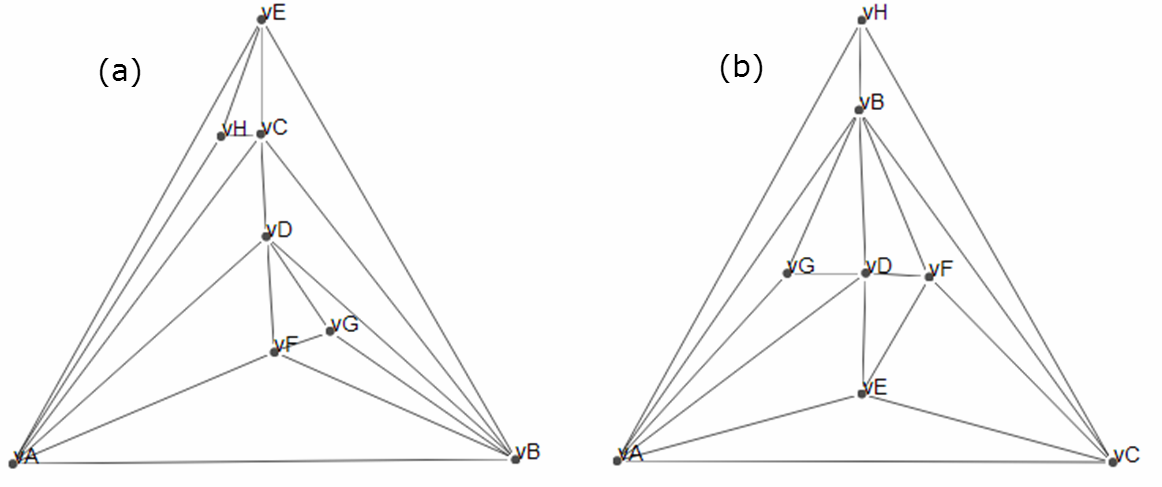} 
\caption{Both graphs have $deg[v_{A}, v_{B}, v_{C}, ..., v_{H}]$=[6, 6, 5, 5, 4, 4, 3, 3] however they are not isomorphic and so are not planar representations of a single graph. Furthermore, they have different numbers of $C_3$ and $C_4$ with the graph shown in Panel (a) having $C_3=16$, $C_4=5$ and the graph shown in Panel (b) having $C_3=14$, $C_4=2$.}
\end{center}
\end{figure} 

\subsection{Standard Spherical Triangulation Form}

We now consider how these different embeddings relate to each other using the idea of diagonal flips, as introduced by Negami \cite{Negami94}.\\

In 1936 Wagner \cite{Wagner36} proved that any two triangulations of the sphere can be transformed into each other by a finite series of diagonal flips (see also Ref. \cite{Bose12}). Although this does not hold for surfaces in general it has been shown to be true for triangulations of the torus, projective plane and Klein bottle. Negami \cite{Negami94} generalized Wagner's result: \\

\begin{theorem} For any closed surface $F^{2}$, there exists a positive integer $N$ such that two triangulations $G$ and $G'$ of $F^{2}$ are equivalent to each other under diagonal flips if $|V(G)| = |V(G')| \geq N$. 
\end{theorem} 

In the case of the PMFG, the triangulation of the sphere, $N=4$. \\

\begin{lemma} Any maximal graph with $n$ vertices, $n \geq 4$, can be transformed to the \textit{standard spherical triangulation} (or normal form), (see Figure 5), using a series of diagonal flips. \end{lemma} (For proof refer to Ref. \cite{Ore67}). \\

From $n \geq 4$, the degrees of each vertex in the standard spherical triangulation are as follows: \\

$deg[v_{1}, v_{2}, v_{3}, v_{4}, ..., v_{i}, ..., v_{n-1}, v_{n}] = [n-1, n-1, 4, ..., 4, ..., 3, 3] $

\begin{figure}[ht]
\begin{center}
\includegraphics[width=.4\textwidth]{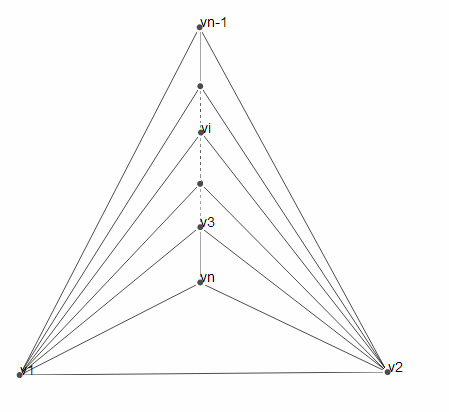}  
\caption{ A maximal graph with $n$ vertices in the standard spherical triangulation.}
\end{center}
\end{figure}

\section{Main Results}
\subsection{Generating Maximal Planar Graphs}

In 1891 Eberhard proposed a system in which a combination of a set of three operations could generate all possible maximally (filtered) planar graphs. We begin with the complete graph $K_4$ and then choose a generating operation, $\varphi_1, \varphi_2, \varphi_3$, from the operation set, $\Phi$. Each generating operation adds a new vertex to the graph. This system is denoted as $<K_4; \Phi = \{{\varphi_1, \varphi_2, \varphi_3 }\}>.$  \\

For $\varphi_1, \varphi_2, \varphi_3$ begin by deleting all of the chords of a pure chord-cycle $\mathcal{C}_k$ with length $k = (3, 4, 5)$ respectively. Then add a new vertex inside $\mathcal{C}$, which is connected to all vertices of $\mathcal{C}$ so that a wheel subgraph is created.\\

Figure 6 shows how a pure chord-cycle transforms under each Eberhard operation.

\begin{figure}[H]
\begin{center}
\includegraphics[width=.9\textwidth]{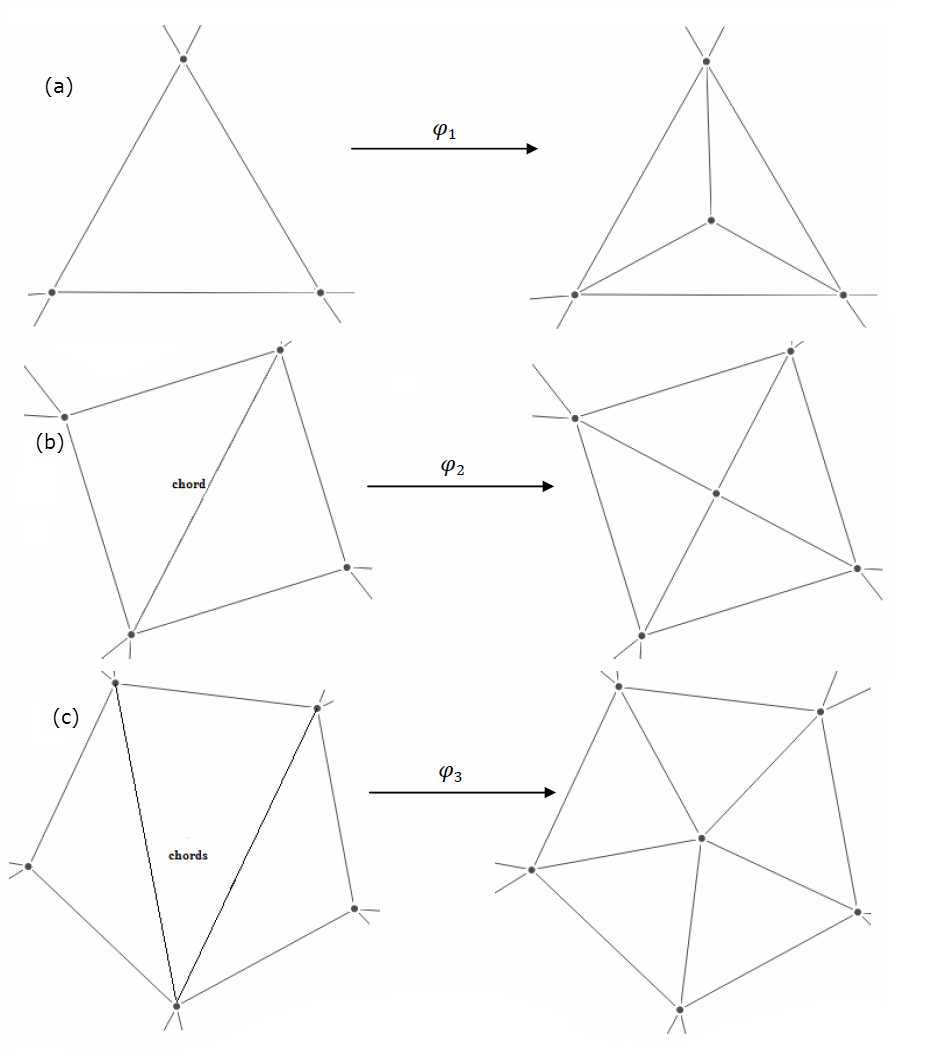}  
\caption{Panel (a) - The first Eberhard operation, $\varphi_1$, Panel (b) - The second Eberhard operation, $\varphi_2$, Panel (c) - The third Eberhard operation, $\varphi_3$}
\end{center}
\end{figure}

\pagebreak

\textbf{Example 2.} For $P_6$ two of the five possible representations are planar. \\

\begin{figure}[H]
\begin{center}
\includegraphics[width=.8\textwidth]{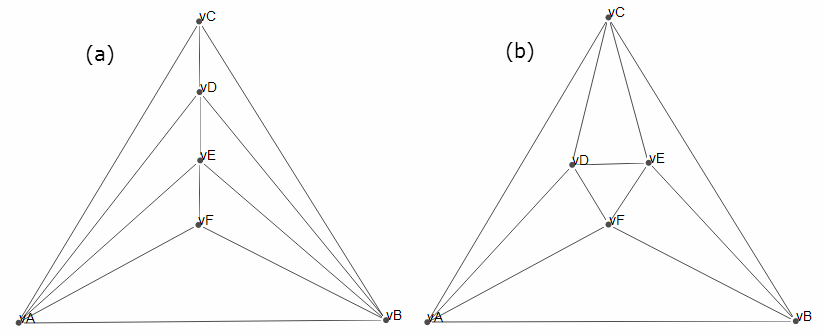}  
\caption{Panel (a) Standard spherical triangulation form, with $C_3= 10$, $C_4 = 3$ and Panel (b) the alternative form, with $C_3=8$, $C_4=0$.}
\end{center}
\end{figure}

Both of these graphs (shown in Figure 7) can be generated using a series of Eberhard's operations, starting with $K_4.$ \\

We first generate $P_5$ from the complete graph $K_4$. There is only one representation of $P_5$ that is planar. \\

\begin{figure}[H]
\begin{center}
\includegraphics[width=.8\textwidth]{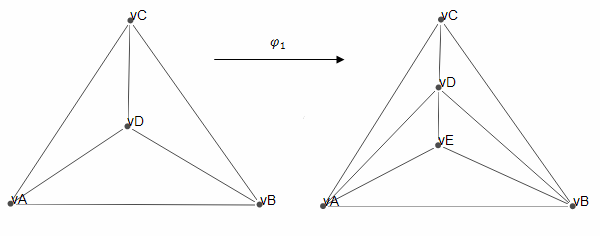}  
\caption{The transformation of $K_4$ to $P_5$ using Eberhard operation $ \varphi_1$.}
\end{center}
\end{figure}

 Then we can generate $P_6$ from $P_5$, (shown in Figure 8), by using any of the three operations from $\Phi$. In $P_5$ there are five pure chord-cycles of length 3 which are as follows: (A, E, C), (A, D, E), (C, D, E), (A, B, D) and (B, C, D). For $\varphi_1$ we add a vertex to the interior of one of the above pure chord-cycles which we join to the three edges of the cycle to create a wheel subgraph and a representation of $P_6$. Note that all $\mathcal{C}_3$ in $P_5$ will generate $P_6$ in the standard spherical triangulation form, (see Figure 9). \\

\begin{figure}[H]
\begin{center}
\includegraphics[width=.8\textwidth]{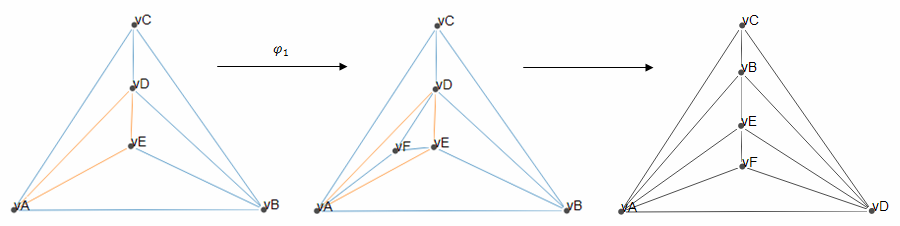}  
\caption{The transformation of $P_5$ to $P_6$ using Eberhard operation $ \varphi_1$.}
\end{center}
\end{figure}

In $P_5$ there are four pure chord-cycles of length 4, all with one chord edge, which are as follows: (A, C, B, D), (A, B, D, E), (A, C, D, E) and (B, C, D, E). For $\varphi_2$ we select one of these pure chord-cycles and start by deleting the chord edge. Then we add a vertex to the interior which we join to the four edges of the cycle to create a wheel subgraph and a representation of $P_6$. Using $\varphi_2$ can generate $P_6$ in both standard spherical triangulation form and the alternative form, depending on which pure chord-cycle is chosen, (see Figure 10).\\

\begin{figure}[H]
\begin{center}
\includegraphics[width=.7\textwidth]{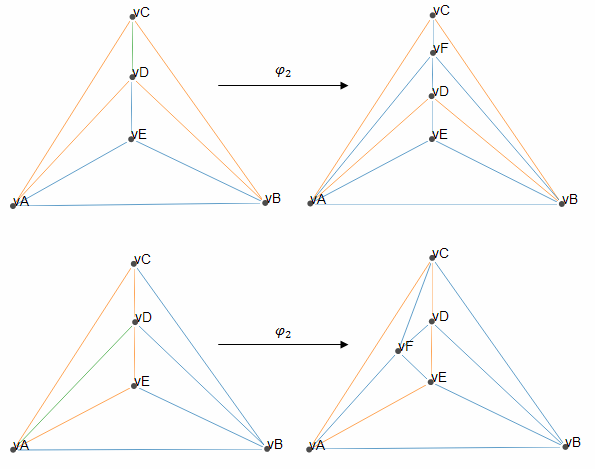}  
\caption{The transformation of $P_5$ to $P_6$ using Eberhard operation $ \varphi_2$.}
\end{center}
\end{figure}

The final option is to use $\varphi_3$ and the one pure chord-cycle of length 5 (A, C, B, D, E). This has two chord edges (A,D) and (C,D) which will be removed before adding a new vertex to the interior. This is then joined to each of the five vertices in the cycle to produce $P_6$ in standard spherical triangulation form,  (see Figure 11).\\

\begin{figure}[H]
\begin{center}
\includegraphics[width=.8\textwidth]{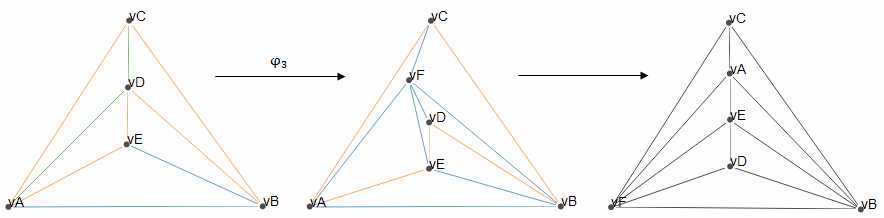}  
\caption{The transformation of $P_5$ to $P_6$ using Eberhard operation $ \varphi_3$.}
\end{center}
\end{figure}

\subsection{Maximum Number of 3- and 4- Cliques}
Using the Eberhard operations to generate maximal planar graphs we can find the maximum number of 3-cliques that will be added during each iteration of the construction algorithm and so consequently the maximum number of 3-cliques that is possible in $P_n$.\\

Now the maximum number of 3-cliques is shown using Theorem 2.

\begin{theorem} Let $P_n$ be a maximal planar graph with $n$ vertices, $n \geq 3$. Then the maximum number of 3-cliques, $C_3^{max}$, that are possible is $C_3^{max}(P_n)=3n-8$.
\end{theorem}

\begin{proof}  With each Eberhard operation there is a new vertex added and also a certain number of 3-cliques - $\varphi_1$ adds three new 3-cliques whereas $\varphi_2$ and $\varphi_3$ both add two new 3-cliques.

Therefore the maximum number of 3-cliques that can be added is three and so we can say that: \\

\begin{equation}
C_3^{max}(P_n) \leq  C_3(K_4) + 3(n - 4),\\
\end{equation}

where $(n - 4)$ is the number of Eberhard operations required to construct $P_n$. As we know that the number of 3-cliques in $K_4$ is always 4 we can simplify (3) and obtain: \\

$C_3^{max}(P_n) \leq 4 + 3n - 12 = 3n - 8$ \\

So the maximum number of 3-cliques possible in $P_n, C_3^{max} = 3n - 8.$ \\
\end{proof}

 We can apply a similar argument to obtain the maximum number of 4-cliques in $P_n$. \\

\begin{theorem} Let $P_n$ be a maximal planar graph with $n$ vertices, $n \geq 3$. Then the maximum number of 4-cliques, $C_4^{max}$, that are possible is $C_4^{max}=n-3$. \\
\end{theorem}
 
\begin{proof} With each Eberhard operation there is a new vertex added, however only $\varphi_1$ adds one new 4-clique, neither $\varphi_2$ and $\varphi_3$ add any new 4-cliques.  \\

Therefore the maximum number of 4-cliques that can be added is one and so we can say that: \\

\begin{equation}
C_4^{max}(P_n) \leq  C_4(K_4) + (n - 4), \\
\end{equation}

where $(n - 4)$ is the number of Eberhard operations required to construct $P_n$. As we know that the number of 4-cliques in $K_4$ is always 1 we can simplify (4) and obtain: \\

$C_4^{max}(P_n) \leq 1 + n - 4 = n - 3$ \\

So the maximum number of 4-cliques possible in $P_n, C_4^{max}= n - 3.$
\end{proof}

\subsection{3- and 4-Cliques in the Standard Spherical Triangulation Form}

As discussed in Section 3 there can be various representations of a graph and the number of 3-cliques that form between the vertices will depend on the structure of the graph. The minimum number of 3-cliques that will form in a PMFG with $n$ vertices = $2n-4$ as it will be equal to the number of surface triangles, including the vertices that form the infinite face (as shown in Proposition 1). From 4.2 we now have an expression for the maximum number of 3-cliques that can form in a PMFG. We now show that the standard spherical triangulation form always contains the maximum number of 3-cliques.\
When a maximal planar graph is in the standard spherical triangulation form we have two types of 3-cliques - firstly those formed by surface triangles and secondly those that enclose 3 surface triangles which share common edges (as shown in Figure 2) which form between a vertex and the two vertices with degree $n-1$. We call these two forms of triangles surface triangles and enclosing triangles respectively.

\begin{theorem} For a maximal graph with $n$ vertices in the standard spherical triangulation form the number of $C_3 = 3n - 8$ and the number of $C_4 = n - 3$.
\end{theorem}

\begin{proof} Number of $C_3$ in the standard spherical triangulation form = (the number of surface triangles) + (the number of enclosing triangles) - (unbounded face) \\
$= (2n-4) + (n-3) - 1 = 3n - 8.$

Note that the number of enclosing triangles is equal to what we have shown to be the maximum number of $C_4$.

\end{proof}

\section{Conclusion}
In this paper we have discussed possible embeddings of $n$-vertex triangulations in the form of maximal planar graphs. We used the generating operations proposed by Eberhard to construct these maximal planar graphs and have proven that the maximum number of 3-cliques that can exist in a maximal planar graph with $n$ vertices is $3n-8$ and the maximum number of 4-cliques that can exist is $n-3$, where the number of vertices $n \geq 4$. This is true for when a maximal planar graph is constructed using the PMFG algorithm.  \\

Finally, we have shown how any maximal planar graph can be transformed to a standard spherical triangulation form retaining the original number of vertices and edges and that this structure will always contain the maximum number of 3- and 4- cliques. \\

The work on PMFG has recently been extended to include Partial Correlation Planar maximally filtered Graphs (PCPG) which use partial correlations between the stocks as opposed to Pearson correlations \cite{Kenett10}, \cite{Mai14}. It would be interesting to see how the use of partial correlations affects the 3- and 4-cliques. We leave this for future work.

\section*{References}

\begin{itemize}
\bibitem [11] {Tumminello05} M. Tumminello, T. Aste, T. Di Matteo, R.N. Mantegna. Proceedings of National Academy of Sciences 102 (2005) 1042.
\bibitem [1] {Newman03} M.E.J Newman. SIAM Review 45 (2003) 167.
\bibitem [2] {Toivonen06} R. Toivonen, J.-P Onnela, J. Saramaki, J. Hyvonen, K. Kaski. Physica A 371 (2006) 851.
\bibitem [3] {Allen00} F. Allen and D. Gale. The Journal of Political Economy 108, 1 (2000) 1.
\bibitem [4] {Bonanno04} G. Bonanno, G. Caldarelli, F. Lillo, S. Micciche, N. Vandewalle, R. N. Mantegna. European Physical Journal B 38 (2004) 363.
\bibitem [5] {Stanley99} H.E. Stanley, R.N. Mantegna. Introduction to Econophysics: Correlations and Complexity in Finance (1999).
\bibitem [6] {Mantegna99} R.N. Mantegna, European Physical Journal B 11 (1999) 193.
\bibitem [7] {Onnela02} J.-P. Onnela,  A. Chakraborti, K. Kaski, J. Kertesz. European Physical Journal B 30 (2002) 285.
\bibitem [8] {Onnela03} J.-P. Onnela, A. Chakraborti, K. Kaski, J. Kertesz, A. Kanto. Physic Scripta T106 (2003) 48.
\bibitem [9] {Tse10} C.K. Tse, J. Liu, F. C. M. Lau. Journal of Empirical Finance 17 (2010) 659.
\bibitem [10] {Qiu10} T. Qiu, B. Zheng, G. Chen. New Journal of Physics 12 (2010).
\bibitem [12] {Hopcroft74} J. Hopcroft and R. Tarjan. Journal of the Association for Computing Machinery 2  4 (1974) 549.
\bibitem [13] {Huang09} W-Q. Huang, X-T. Zhuang, S. Yao. Physica A 388 (2009) 2956.
\bibitem [14] {Pozzi13} F. Pozzi, T. Di Matteo, T. Aste. Scientific Reports 3 (2013) 1665.
\bibitem [15] {Eryigit09} M. Eryigit and R. Eryigit. Physica A 388 (2009) 3551.
\bibitem [16] {Song12} W-M. Song, T. Di Matteo, T. Aste. PLoS ONE 7(3): e31929 (2012).
\bibitem [17] {Aste05a} T. Aste, T. Di Matteo and S.T. Hyde. Physica A 346 (1-2 Spec. Iss.) (2005).
\bibitem [18] {Aste05b} T. Aste and T. Di Matteo. Proc. SPIE 5848, Noise and Fluctuations in Econophysics and Finance, 100 (2005).
\bibitem [19] {Tumminello07} M. Tumminello, T. Di Matteo, T. Aste, R.N. Mantegna. European Physical Journal B 55 (2007) 209.
\bibitem [20] {Coronnello05} C. Coronnello, M. Tumminello, F. Lillo, S. Micciche and  R.N. Mantegna. Acta Phys. Pol. B 36 (2005) 2653.
\bibitem [21] {Negami94} S. Negami. Discrete Mathematics 135 (1994) 225.
\bibitem  [22] {Wilson85} R. J. Wilson. Introduction to Graph Theory. Longman Group Ltd. (1985).
\bibitem  [23] {Wagner36} K. Wagner. Journal der Deutschen Mathematiker-Vereinigung 46 (1936) 26.
\bibitem [24] {Bose12} P. Bose, S. Verdonschot. Computational Geometry (2012) 29.
\bibitem [25] {Ore67} O. Ore. The Four-Color Problem. Academic Press (1967) 9.
\bibitem [26] {Kenett10} D.Y. Kenett, M. Tumminello, A. Madi, G. Gur-Gershgoren, R.N. Mantegna, E. Ben Jacob. Public Library of Science 5, 12 (2010) e15032.
\bibitem [27] {Mai14} Y. Mai, H. Chen, L. Meng. Physica A 396 (2014) 235.
\end{itemize}

All figures created using FNA.fi.

\end{document}